\documentclass[12pt]{article}
\usepackage{amsmath}
\usepackage{graphicx}%
\usepackage{amsfonts}%
\usepackage{amssymb}
\usepackage{overpic}
\usepackage{subfigure}
\usepackage{rotating}
\usepackage{url}
\usepackage{fullpage}
\usepackage{comment}
\usepackage{enumerate}
\usepackage{multirow}
\usepackage{booktabs}
\usepackage{tabularx}
\usepackage{makecell}
\usepackage{booktabs}
\usepackage{natbib}
\usepackage{diagbox}
\usepackage{caption}

\usepackage[colorlinks = true,
            linkcolor = blue,
            urlcolor  = blue,
            citecolor = blue,
            anchorcolor = blue]{hyperref}
\usepackage[letterpaper,left=1in,top=1in,right=1in,bottom=1in,lines=25]{geometry}

\usepackage{setspace}
\newcommand{\bI}{ {\boldsymbol I} }

\newcommand{\bw}{ {\boldsymbol w} }

\newcommand{\bx}{ {\boldsymbol x} }
\newcommand{\bX}{ {\boldsymbol X} }
\newcommand{\by}{ {\boldsymbol y} }

\newcommand{\bz}{ {\boldsymbol z} }

\newcommand{\bbeta}{ {\boldsymbol \beta} }

\newcommand{\bgamma}{ {\boldsymbol \gamma} }

\newcommand{\bet}{ {\boldsymbol \eta} }

\newcommand{\btheta}{ {\boldsymbol \theta} }

\newcommand{\bzero}{ {\boldsymbol 0} }

\newcommand{\given}{\,|\,}

\newtheorem{theorem}{Theorem}[section]

\newenvironment{proof}[1][Proof]{\begin{trivlist}
\item[\hskip \labelsep {\bfseries #1}]}{\end{trivlist}}

\newcommand{\qed}{\nobreak \ifvmode \relax \else
      \ifdim\lastskip<1.5em \hskip-\lastskip
      \hskip1.5em plus0em minus0.5em \fi \nobreak
      \vrule height0.75em width0.5em depth0.25em\fi}

\title{Bayesian Causal Inference with Bipartite Record Linkage}
\author{Sharmistha Guha, Jerome P. Reiter and Andrea Mercatanti}
\begin{document}
\maketitle

\begin{abstract}
In many scenarios, the observational data needed for causal inferences are spread over two data files. In particular, we consider scenarios where one file includes covariates and the treatment measured on one set of individuals, and a second file includes responses measured on another, partially overlapping set of individuals. 
In the absence of error free direct identifiers like social security numbers, straightforward merging of separate files is not feasible, so that records must be linked using error-prone variables such as names, birth dates, and demographic characteristics.  Typical practice  in such situations generally follows a two-stage procedure: first link the two files using a probabilistic linkage technique, then make causal inferences with the linked dataset.  This does not propagate uncertainty due to imperfect linkages to the causal inference, nor does it leverage relationships among the study variables to improve the quality of the linkages. We propose a hierarchical model for simultaneous Bayesian inference on probabilistic linkage and causal effects that addresses these deficiencies. 
Using simulation studies and theoretical arguments, we show the hierarchical model can improve the accuracy of estimated treatment effects, as well as the record linkages, compared to the two-stage modeling option. We illustrate the hierarchical model using a causal study of the effects of debit card possession on household spending. 
\end{abstract}


\section{Introduction}
Often, researchers seek to make causal inferences from variables spread over two datasets.
For example, a social scientist seeks to link records from a survey and an administrative database to assess the effect of some policy  on economic outcomes. 
Similarly, a health researcher seeks to link patients'  electronic health records and Medicare claims data to assess the effect of some medical intervention.  As a final example, a researcher seeks to link records from a study done in the past to records in a current database to make inferences about long-term effects of a treatment, without having to incur the substantial costs of collecting new primary data. 

When perfectly measured, unique identifiers like social security numbers or Medicare patient IDs are available in the two files, it is reasonably straightforward to  link individuals across the files (based on these identifiers). However, often direct identifiers are missing from one or more files, or may not be made available due to privacy restrictions. 
In such situations, data files have to be linked based on indirect identifiers, such as individuals' names, birth dates, addresses, and demographic information. 
These are inherently imperfect, e.g., they could be recorded differently on the files. This introduces uncertainty in linkages that should be propagated to the causal inferences. 

Historically, record linkage and causal inference have been carried out as a two-stage process. The researcher first links records using a probabilistic record linkage model based on indirect identifiers, not taking into account available information on the outcome, covariate or treatment status. Subsequently, the researcher uses the set of linked records in a causal inference procedure. This two-stage approach suffers from two drawbacks.  First, it does not propagate uncertainty from imperfect linkages.   Second, it does not take advantage of relationships among the study variables that could enhance the accuracy of the linkages.  

In this article, we propose a Bayesian hierarchical modeling framework for simultaneous causal inference  and record linkage in observational studies. In particular, we consider scenarios where one file includes the treatment indicator and causally relevant covariates measured on a set of individuals, and the other file includes outcomes measured on a partially overlapping set of individuals.  We follow the Bayesian paradigm for causal inference and posit models for the missing potential outcomes, conditional on the linking status and known covariates.  We couple these outcome models with a probabilistic model for the unknown linkage statuses, i.e., which record pairs are links and which are not.
For the outcome models, we consider both parametric and semi-parametric forms, with the latter based on a regression of the outcome on a flexible function of the propensity scores \citep{rosenbaum1983central}.  For the record linkage model, we use the Bayesian version of the \cite{fellegi1969theory} model proposed by \cite{sadinle2017bayesian}.
As part of the model estimation, we generate plausible values of the missing potential outcomes, which we then use to estimate posterior distributions of causal effects.

Our work adds to a body of literature that uses Bayesian methods for simultaneous record linkage and statistical inference, including regression modeling \citep{gutman2013bayesian, dalzell2018regression} and  population size estimation \citep{domingo2011privacy, tancredi2011hierarchical,  sadinle2018bayesian, tancredi2018unified}. 
It also adds to the literature on non-Bayesian methods for simultaneous record linkage and estimation \citep[e.g., ][]{scheuren1991regression, lahiri2005regression, chipperfield2011maximum,  solomon2019}. None of these Bayesian and and non-Bayesian works consider causal inference as the analysis goal.    \cite{wortman2018simultaneous} introduced the concept of allowing the causal model to inform the linkage model.  
Their (non-Bayesian) approach uses point estimates of average causal effects to determine the thresholds  at which record pairs links are declared links in a \cite{fellegi1969theory} algorithm.  It does not use the causal estimates to determine the record pairs to consider as possible links in the first place, which our Bayesian approach does. 
Further, their approach does not provide uncertainty quantification.

The remainder of the article proceeds as follows. In Section \ref{sec1} we discuss the background, notations and the formulation of the Bayesian hierarchical model. We also present theoretical results arguing for improved inference on record linkage from a joint model compared to a two-stage model. In Section~\ref{post_comp} we  describe posterior computation for the model. In Section \ref{secsim}, we provide results from simulation studies used to assess the effectiveness of the hierarchical model, both for causal inference and for linkage quality. 
In Section \ref{rda} we
apply a hierarchical model to data from an Italian household survey to link records between different files and assess the effect of debit card possession on household spending. Finally, in Section \ref{conclusion} we conclude with an eye towards future work.

\section{Model and Prior Formulation}\label{sec1}
We define a few key concepts and assumptions related to causal inference in Section~\ref{bg}, and describe probabilistic record linkage  in Section~\ref{RL}. We propose the Bayesian hierarchical modeling approach in Section~\ref{RL-CI}.
\subsection{Background and Notation for Bayesian Causal Inference}\label{bg}
We assume a binary treatment, $w_i\in\{0,1\}$, with $w_i=1$ and $w_i=0$ indicating treatment and control assignment to individual $i$, respectively. Let $\bx_i$ be the $p\times 1$ covariate vector and $y_i$ be the (continuous) outcome for individual $i$. 
Each individual is assumed to have two potential outcomes \citep{rubin1974estimating}, one under each value of the treatment. We denote $y_i(1)$ and $y_i(0)$ as the potential outcomes for individual $i$ when $w_i=1$ or $w_i=0$, respectively. The treatment effect for the $i$th individual is given by $T_i=y_i(1)-y_i(0)$.  Other treatment effects can be defined as well, such as $y_i(1)/y_i(0)$, although here we consider effects in the form of $T_i$. 

In reality, for each individual $i$, we can observe only one of $y_i(1)$ and $y_i(0)$; that is, we can observe $y_i=w_i y_i(1)+(1-w_i)y_i(0)$. Bayesian approaches to causal inference essentially treat the unobserved potential outcomes as missing data \citep{rubin2005bayesian, hill2011bayesian, ding2018causal}.  One can impute the missing values repeatedly by sampling from posterior predictive distributions, and use the resulting draws of each $T_i$ to make statements about causal effects.  For example, one can compute the posterior distribution of the average of the causal effects for the $n$ individuals in the study, $\bar{T} = \sum_{i}T_i/n$.

Following convention, we make the following assumptions to facilitate causal inferences.
\begin{enumerate}
\item \emph{Stable unit treatment value assumption} (SUTVA): The SUTVA contains two sub-assumptions, no interference between units (i.e., the treatment applied to one unit does not affect the outcome for another unit) and no different versions of any treatment \citep{rubin1974estimating}.
\item \emph{Strong ignorability}: Strong ignorability stipulates that $(y_i(0),y_i(1))\perp w_i|\bx_i$ for all $i$, which means that there is no unobserved confounding, and that $0<P(w_i=1|\bx_i)<1$.
\end{enumerate}

We also make use of propensity scores \citep{rosenbaum1983central}. For any individual $i$, let the propensity score $e(\bx_i)=P(w_i=1|\bx_i),$ i.e., the probability of being assigned to treatment given the covariate $\bx_i$.  \cite{rosenbaum1983central} show that the treatment assignment is independent of $\bx_i$ given $e(\bx_i)$ under strong ignorability. Typically, propensity scores are estimated using binary regressions of the treatment on causally-relevant covariates.  Analysts can use the resulting estimate in a variety of ways, for example, to create subsets of matched treated and control records \citep{stuart2010matching}.

\subsection{Background and Notation for Probabilistic Record Linkage}\label{RL}

We consider the scenario where we seek to link two files, 
File A and File B, comprising $n_A$ and $n_B$ records, respectively. Without loss of generality, we assume that $n_A\geq n_B$. We suppose each individual or entity is recorded at most once within each file, i.e., each file contains no duplicates. 
Under this setting, the goal of record linkage is to identify which records in File A and File B refer to the same subject. This setting is known as bipartite record linkage \citep{sadinle2017bayesian}.

A corollary to the no-duplicates assumption comes in the form of a maximum one-to-one restriction in the linkage, i.e., a record in one file can be linked with a maximum of one record in the other file. Most commonly, the one-to-one linkage is enforced as a post-processing step after identifying  a set of potentially many-to-one links.  
\citep[e.g., ][]{fellegi1969theory,jaro1989advances,winkler1993improved,belin1995method, larsen2001iterative, herzog2007data}.  
Alternatively, one can embed the bipartite matching constraint into a Bayesian model \citep{fortini2002modelling,tancredi2011hierarchical,larsen2010record,gutman2013bayesian, sadinle2017bayesian, dalzell2018regression}, as we do here. 

Following \cite{sadinle2017bayesian}, we introduce $\bz = (z_1, \dots, z_{n_B})'$ for the records in File B to encode a particular linking status between the two files. Specifically, let 
\begin{align*}
z_j=\left\{\begin{array}{cc}
i, & \mbox{if record $i$ in File A and record $j$ in File B belong to same individual}\\
n_A+j, & \mbox{if record $j$ in File B has no link in File A}
\end{array}\right.
\end{align*}
In the context of bipartite matching, one enforces $z_j\neq z_{j'}$ whenever $j\neq j'$.

Suppose the two files include $F$ variables in common that can be used to link records across the files.  We call these the linking variables or fields. For each pair of records $(i,j)$ in File A$\times$File B, we define a $F$-dimensional vector  
$\bgamma_{ij}=(\gamma_{1,ij}, \dots, \gamma_{F,ij})'$, where $\gamma_{f,ij}$ is a score reflecting the similarity in  field $f$ for the record pair. In this article, 
for  
nominal variables (e.g., age, birth year, sex), we set $\gamma_{f,ij}=1$ when the values of field $f$ for records $i$ and $j$ are equal, and set $\gamma_{f,ij}=0$ otherwise. For string fields (e.g., names) we  take into account partial agreement using the normalized Levenshtein Similarity metric  \citep{winkler1990string}. This metric ranges between $0$ (no agreement) and $1$ (full agreement).  We obtain it using the ``levenshteinSim'' function in the \texttt{RecordLinkage} package in \texttt{R}. We convert the distances into a binary  $\gamma_{f,ij}$ by setting $\gamma_{f,ij} = 1$ when the distance metric exceeds a predetermined threshold (e.g., 0.95), and $\gamma_{f,ij} = 0$ otherwise.  One can convert the metric into a multinomial variable for more refined comparisons \citep{sadinle2018bayesian, wortman2018simultaneous}. 
 
Following \cite{fellegi1969theory} and related literature, we assume that $\bgamma_{ij}$ is a random realization from a mixture of two distributions, one for true links and the other for nonlinks.  We have 
\begin{align}\label{prob_rec_link}
\bgamma_{ij} | (z_j=i)  \stackrel{iid}{\sim} g(\btheta_{m}),\:\:\:
\bgamma_{ij} | (z_j\neq i) \stackrel{iid}{\sim} g(\btheta_{u}),
\end{align}
where $\btheta_m=(\theta_{1,m}, \dots, \theta_{F,m})'$ and $\btheta_u=(\theta_{1,u}, \dots, \theta_{F,u})'$ comprise probabilities of agreement for each field specific to each mixture component. 
Again following typical practice, for computational convenience we assume conditional independence across fields, so that 
\begin{align}\label{eq:link_var}
g(\btheta_m)=P(\bgamma_{ij}|z_j=i)=\prod_{f=1}^F P(\gamma_{f,ij}|z_j=i)=\prod_{f=1}^F \theta_{f,m}^{\gamma_{f,ij}}
(1-\theta_{f,m})^{1-\gamma_{f,ij}}\nonumber\\
g(\btheta_u)=P(\bgamma_{ij}|z_j\neq i)=\prod_{f=1}^F P(\gamma_{f,ij}|z_j\neq i)=\prod_{f=1}^F \theta_{f,u}^{\gamma_{f,ij}}
(1-\theta_{f,u})^{1-\gamma_{f,ij}}.
\end{align}

As a prior distribution on the set of $z_j$, such that $z_j\neq z_{j'}$ for any $j\neq j'$, we follow a construct used in the bipartite record linkage literature, including \cite{fortini2002modelling}, \cite{larsen2010record} and \cite{sadinle2017bayesian}. Specifically, let $I(z_j\leq n_A)\sim Ber(\pi)$, where $\pi$ represents the proportion of matches expected a priori as a fraction of the smaller file.  Here and throughout, $I(\mathcal{E})=1$ when its argument $\mathcal{E}$ is true, and $I(\mathcal{E})=0$ otherwise. We assume $\pi$ is distributed according to a 
Beta($\alpha_{\pi},\beta_{\pi}$) a priori. Marginalizing over $\pi$, the total number of links between File A and File B, given by $n_{AB}(\bz)=\sum_{j=1}^{n_B}I(z_j\leq n_A)$, is distributed according to a Beta-binomial ($n_B,\alpha_{\pi},\beta_{\pi}$) distribution. Conditioning on the number of records in File B with a link, all possible bipartite pairings are taken as equally likely. The final form of the prior distribution of $\bz$, marginalizing over $\pi$, is given by 
\begin{align}\label{prior_z}
P(\bz | a_{\pi}, b_{\pi}) = \frac{(n_A - n_{AB}(\bz))!}{n_A!} \frac{B(n_{AB}(\bz)+\alpha_{\pi}, n_B - n_{AB}(\bz)+\beta_{\pi})}{B(\alpha_{\pi},\beta_{\pi})}. 
\end{align}
The choice of the hyper-parameters $\alpha_{\pi}$ and $\beta_{\pi}$ provides prior information on the number of overlapping records between the two files. We discuss the specific choices of $\alpha_{\pi}$ and $\beta_{\pi}$ in Section~\ref{post_comp}. Finally, the parameters $\theta_{f,m}$ and $\theta_{f,u}$ follow i.i.d. $Beta(a=1,b=1)$ distributions for all $f=1, \dots, F.$  

\subsection{Hierarchical Model for Bayesian Causal Inference and Record Linkage}\label{RL-CI}
To develop a joint model for Bayesian causal inference and record linkage, we specify the distribution of the outcome $y_i$ in File A depending on whether or not it is linked to any covariate and treatment in File B. For linked records, 
we specify the conditional distribution of 
$y_i|(\bx_j,w_j)$ using a regression of our choice. For records without a link, we specify 
a model for the marginal distribution of $y_i$.  We couple these with the model for record linkage in (\ref{prob_rec_link}) -- (\ref{prior_z}).

More specifically, the contribution to the likelihood function from the $i$th record in File A is $f_1(y_i \given \bx_j, w_j, \btheta_c)$ when $z_j = i$, and is $f_2(y_i \given \btheta_d)$ when $z_j \neq i$, for any $j$. Here, $\btheta_c$ and $\btheta_d$ represent parameters in the regression and in the marginal model for outcomes, respectively. Let $\by = (y_1, \dots, y_{n_A})'$ and 
$\bw = (w_1, \dots, w_{n_B})'$ be the $n_A \times 1$ vector of outcomes in File A and $n_B \times 1$ vector of treatment statuses in File B, respectively, and $\bX=[\bx_1':\cdots:\bx_{n_B}']'$ be an $n_B\times p$ dimensional matrix of covariates obtained from File B. The joint likelihood is given by  
\begin{align}\label{jm}
& L(\btheta_d,\btheta_c,\btheta_m,\btheta_u,\bz|\{\bgamma_{ij}:1\leq i\leq n_A, 1\leq j\leq n_B\},\by,\bw,\bX)\nonumber \\
& \propto \prod_{\substack{(i,j):\\ z_j = i}} f_1(y_i \given \bx_j, w_j,\btheta_c) 
\times \prod_{\substack{i:z_j \neq i \\ \forall j}}  f_2(y_i|\btheta_d)\nonumber \\
&\times \prod_{i,j} \left \{\prod_{f=1}^{F} {\theta_{f,m}}^{\gamma_{f,ij}} {(1-\theta_{f,m})}^{1-\gamma_{f,ij}} \right\} ^{I(z_j = i)} 
\times  \left\{\prod_{f=1}^{F} {\theta_{f,u}}^{\gamma_{f,ij}} {(1-\theta_{f,u})}^{1-\gamma_{f,ij}}\right\} ^{I(z_j \neq i)}  \nonumber \\
& \times I(z_j\neq z_{j'},\:\mbox{whenever}\:j\neq j').
\end{align}


To illustrate the potential benefit of joint modeling over two-stage modeling, we 
examine the likelihood ratio that any pair of records is linked versus not linked. 
Under the joint model, the likelihood ratio of $i\sim j$ (i.e., record $i$ is linked to record $j$) and $i\not\sim j$ is given by
\begin{align}\label{joint_ratio}
\mbox{Ratio}_{Joint} = 
\frac{L(\btheta_d,\btheta_c,\btheta_m,\btheta_u,\bz|\{\bgamma_{ij}:1\leq i\leq n_A, 1\leq j\leq n_B\},\by,\bw,\bX, i \sim j)}{L(\btheta_d,\btheta_c,\btheta_m,\btheta_u,\bz | \{\bgamma_{ij}:1\leq i\leq n_A, 1\leq j\leq n_B\},\by,\bw,\bX, i \not \sim j)}
\end{align}
Notably, \eqref{joint_ratio} depends on a contribution to the likelihood from the outcome model. In contrast, the likelihood ratio for linking records in the traditional two-stage model 
only involves the likelihood from the assumed probabilistic record linkage model. For this model, we have 
\begin{align}\label{2stage_ratio}
\mbox{Ratio}_{2 Stage} = \prod_{f=1}^{F} \left(\frac{\theta_{f,m}}{\theta_{f,u}}\right)^{\gamma_{f,ij}}
\left(\frac{1-\theta_{f,m}}{1-\theta_{f,u}}\right)^{1-\gamma_{f,ij}}.
\end{align}
Theorem~\ref{th2.1} offers insight into the behavior of $\mbox{Ratio}_{Joint}$ and $\mbox{Ratio}_{2 Stage}$. 
\begin{theorem}\label{th2.1}
Assuming $\frac{f_1(y_i|\bx_j,w_j,\btheta_c)}{f_2(y_i|\btheta_d)}$ is bounded away from $0$ and $\infty$ in its support, we have\\
(a) $E_{i\sim j}[\mbox{Ratio}_{Joint}]\geq E_{i\sim j}[\mbox{Ratio}_{2 Stage}]$\\
(b) $E_{i\not\sim j}[\mbox{Ratio}_{Joint}]\leq E_{i\not\sim j}[\mbox{Ratio}_{2 Stage}]$.
\end{theorem}

\begin{proof}
The likelihood ratio under the joint model, $\mbox{Ratio}_{Joint}$, can be expressed as
\begin{align}
& \frac{\prod_{\substack{(k,l):z_k = l,k \neq j\\l \neq i}} f_1(y_l | x_k, w_k, \btheta_c) \times f_1(y_i | x_j, w_j, \btheta_c) \times \prod_{l:z_k \neq l} f_2(y_l | \btheta_d)}
{\prod_{\substack{(k,l):z_k = l,k \neq j\\l \neq i}} f_1(y_l | x_k, w_k, \btheta_c) \times f_2(y_i | \btheta_d) \times \prod_{l:z_k \neq l} f_2(y_l | \btheta_d)}\nonumber\\
&\qquad\times\frac{\prod_{f=1}^{F} {\theta_{f,m}}^{\gamma_{f,ij}} {(1-\theta_{f,m})}^{1-\gamma_{f,ij}}}
{\prod_{f=1}^{F} {\theta_{f,u}}^{\gamma_{f,ij}} {(1-\theta_{f,u})}^{1-\gamma_{f,ij}}}\nonumber\\
&= \frac{f_1(y_i | x_j, w_j, \btheta_c)}{f_2(y_i | \btheta_d)} 
\prod_{f=1}^{F} \left(\frac{\theta_{f,m}}{\theta_{f,u}}\right)^{\gamma_{f,ij}}
\left(\frac{1-\theta_{f,m}}{1-\theta_{f,u}}\right)^{1-\gamma_{f,ij}}.
\end{align}

The likelihood ratio under the two-stage model, $\mbox{Ratio}_{2 Stage}$,  is given by \eqref{2stage_ratio}, which we abbreviate as 
$h(\theta_{f,m},\theta_{f,u})$. 
Thus,
$log(\mbox{Ratio}_{Joint}) = log(h(\theta_{f,m},\theta_{f,u})) 
+ log\left[\frac{f_1(y_i | x_j, w_j, \btheta_c)}{f_2(y_i|\btheta_d)}\right]$,
and 
$log(\mbox{Ratio}_{2 Stage}) = log(h(\theta_{f,m},\theta_{f,u})).$ 
Therefore, we have 
\begin{align*}
&E_{i \sim j}[log(\mbox{Ratio}_{Joint}) - log(\mbox{Ratio}_{2 Stage})] \\
&= \int \int \left[\prod_{f=1}^{F} {\theta_{f,m}}^{\gamma_{f,ij}} {(1-\theta_{f,m})}^{1-\gamma_{f,ij}}\right] f_1(y_i | x_j, w_j, \btheta_c)\:\: log\left[\frac{f_1(y_i | x_j, w_j, \btheta_c)}{f_2(y_i|\btheta_d)}\right] \geq 0
\end{align*}
as a consequence of this expression being a Kullback-Leibler divergence. And, we have 
\begin{align*}
&E_{i \not \sim j}[log(\mbox{Ratio}_{Joint}) - log(\mbox{Ratio}_{2 Stage})]\\
&= \int \int \left[\prod_{f=1}^{F} {\theta_{f,u}}^{\gamma_{f,ij}} {(1-\theta_{f,u})}^{1-\gamma_{f,ij}}\right] f_2(y_i|\btheta_d)\:\: log\left[\frac{f_1(y_i | x_j, w_j, \btheta_c)}{f_2(y_i|\btheta_d)}\right] \leq 0,
\end{align*}
where the last inequality follows from the fact that the expression is $(-1)$ times the Kullback-Leibler divergence between the two densities $f_1$ and $f_2$.
\end{proof}
Theorem \ref{th2.1} indicates that the likelihood ratio for the joint model is more extreme than the likelihood ratio for the two stage model, which facilitates more accurate identification of a link or no link between records $i$ and $j$. 

\subsubsection{Outcome Models}\label{sec:causla_model}
Naturally, one should specify $f_1(y_i|w_j,\bx_j,\btheta_c)$ and $f_2(y_i|\btheta_d)$ to describe the distribution of outcomes as faithfully as possible.  In this article, we specify models for $y_i\in\mathcal{R}$ and assume  $f_2(y_i|\btheta_d)=N(y_i|\mu_1,\sigma_1^2)$; 
setting a more complicated distributional form for $f_2$ or extending to a categorical $y_i$ is relatively straightforward. 
For $f_1(y_i|w_j,\bx_j,\btheta_c)$, we use   
a general mean-zero additive error form, 
\begin{align}\label{mean0}
y_i = m(\bx_j, w_j) + \epsilon_{i,j},\:\:\:\epsilon_{i,j} \sim N(0,\sigma^2).
\end{align}
The specification for $m(\bx_j,w_j)$ could be a simple linear form, although often in observational studies it is advantageous to use more flexible modeling \citep{hill2011bayesian}.  

We use a computationally favorable yet flexible specification for $m(\bx_j,w_j)$.  In particular, we assume  
\begin{align}\label{mean_causal}
m(\bx_j,w_j)= m_1(\hat{e}(\bx_j)) + m_2(\hat{e}(\bx_j)) w_j,
\end{align}
where $\hat{e}(\bx_j)$ is the estimated propensity score. In the simulations of Section~\ref{secsim}, we use $\hat{e}(\bx_j)=\textcolor{blue}{\rho}^{-1}(\bx_j'\hat{\bet})$, where $\rho(\cdot)$ is the logit link function and $\hat{\bet}$ is the maximum likelihood estimate of $\bet$ obtained by fitting a logistic regression of $w_j$ on $\bx_j$ for all $j \in \mathcal{B} = \left\{j: z_j=i, \:\mbox{for some}\: i, 1\leq i\leq n_A\right\}$. 

To afford model flexibility, we propose a semi-parametric choice for $m_1(\cdot)$ and $m_2(\cdot)$ 
using penalized splines \citep{ruppert2003semiparametric}. Let $\kappa_1<\kappa_2<\cdots<\kappa_m$ be a set of $m$ fixed knot points in $(0,1)$. The functions $m_1(\cdot)$ and $m_2(\cdot)$ are represented using spline basis functions,
\begin{align}\label{nonpar}
m_1(\hat{e}(\bx_j)) &=\beta_0+\sum_{l_1=1}^{s}\beta_{l_1}\hat{e}(\bx_j)^{l_1}+\sum_{l_2=1}^{m}\beta_{s+l_2}(\hat{e}(\bx_j)-\kappa_{l_2})_{+}^{l_2}\nonumber\\
m_2(\hat{e}(\bx_j)) &=\sum_{l_1=1}^{s}\gamma_{l_1}\hat{e}(\bx_j)^{l_1}+\sum_{l_2=1}^{m}\gamma_{s+l_2}(\hat{e}(\bx_j)-\kappa_{l_2})_{+}^{l_2}.
\end{align}
So, the parameters are $\btheta_c=(\beta_0, \beta_1,  \dots,\beta_{s},\beta_{s+1},\dots, \beta_{s+m},\gamma_1, \dots, \gamma_s,\gamma_{s+1},\dots,\gamma_{s+m},\sigma^2)'.$
 This modeling framework is motivated by the penalized spline regression approaches in the Bayesian survey sampling literature \citep{zheng2003penalized,zheng2005inference}, with survey weights replaced by propensity scores.
 
 We suggest placing a large number of knots to estimate the semi-parametric functions accurately. However, even a moderately large choice of $m$ may result in model over-fitting. We therefore regularize the spline coefficients $\beta_{s+1}, \dots, \beta_{s+m}$ and $\gamma_{s+1}, \dots, \gamma_{s+m}$. 
 using Bayesian Lasso shrinkage priors. 
 Following \cite{park2008bayesian}, a scale-mixture representation of the Bayesian Lasso shrinkage prior is given by  
\begin{align}
&\beta_{s+k}|\tau_{1,k}^2\sim N(0,\sigma^2\tau_{1,k}^2), \:\:\:\gamma_{s+k}|\tau_{2,k}^2\sim N(0,\sigma^2\tau_{2,k}^2)\nonumber\\
&\tau_{1,k}^2\stackrel{iid}{\sim} Exp(\lambda_1^2),\:\:\tau_{2,k}^2\stackrel{iid}{\sim} Exp(\lambda_2^2),\:\:k=1, \dots, m\nonumber\\
& \lambda_1^2\sim Gamma(r_1,\delta_1),\:\:\lambda_2^2\sim Gamma(r_2,\delta_2).
\end{align}
 We assign  $\beta_1,...,\beta_s\stackrel{i.i.d.}{\sim}N(0,1)$ and $\gamma_1,...,\gamma_s\stackrel{i.i.d.}{\sim} N(0,1)$ priors. We also assign $\beta_0\sim N(0,1)$ and $\sigma^2\sim IG(a_{\sigma},b_{\sigma})$ priors. The prior specification is completed by setting prior distributions on $\btheta_d=(\mu_1,\sigma_1^2)'$ as $\mu_1 \sim N(0,1)$ and $\sigma_1^2 \sim IG(a_{\sigma_1},b_{\sigma_1})$ a priori.  We discuss the choice of hyper-parameters further in Section~\ref{post_comp}.

For comparisons, we also consider a parametric outcome regression.
In this model, we set $m_1(\hat{e}(\bx_j))=\beta_0+\hat{e}(\bx_j)\beta_1$ and $m_2(\hat{e}(\bx_j))=\alpha$, so that $\btheta_c=(\beta_0,\beta_1,\alpha,\sigma^2)'$. 
We assign $\bbeta=(\beta_0,\beta_1)'$ and $\alpha$ a multivariate normal prior distribution.  We let  $(\bbeta,\alpha)'\sim N(\bzero,\bI)$, and let $\sigma^2$ follow an IG(\textcolor{blue}{$a_{\sigma},b_{\sigma}$}) prior.

\section{Posterior Computation}\label{post_comp}
Incorporating the prior information, the full posterior for the model with the semi-parametric outcome regression is proportional to
\begin{align}
& L(\btheta_d,\btheta_c,\btheta_m,\btheta_u,\bz|\{\bgamma_{ij}:1\leq i\leq n_A, 1\leq j\leq n_B\},\by,\bw,\bX)\times P(\bz \given {\alpha}_{\pi} ,{\beta}_{\pi})\nonumber\\
&\times \prod_{f=1}^F \theta_{f,m}^{a-1} (1 - \theta_{f,m})^{b-1}
 \times \prod_{f=1}^F \theta_{f,u}^{a-1} (1 - \theta_{f,u})^{b-1} \times IG(\sigma^2|a_{\sigma},b_{\sigma})\times N(\beta_0|0,1)\nonumber\\
& \times\prod_{k=1}^s N((\beta_k,\gamma_k)'|0,I)
 \times\prod_{k=1}^m \left[N(\beta_{k+s}|0,\sigma^2\tau_{1,k}^2)\times N(\gamma_{k+s}|0,\sigma^2\tau_{2,k}^2)\right]\nonumber\\
&\times \prod_{k=1}^m \left[ Exp(\tau_{1,k}^2|\lambda_1^2)\times Exp(\tau_{2,k}^2|\lambda_2^2)\right]
\times Gamma(\lambda_1^2|r_1,\delta_1)\times Gamma(\lambda_2^2|r_2,\delta_2)\nonumber\\
&\times N(\mu_1|0,1) \times IG(\sigma_1^2|a_{\sigma_1},b_{\sigma_1}).
\end{align}
Similarly, the full posterior for the model with the parametric outcome regression is proportional to
\begin{align}
& L(\btheta_d,\btheta_c,\btheta_m,\btheta_u,\bz|\{\bgamma_{ij}:1\leq i\leq n_A, 1\leq j\leq n_B\},\by,\bw,\bX)  \times P(\bz \given {\alpha}_{\pi} ,{\beta}_{\pi})\nonumber\\
&\times \prod_{f=1}^F \theta_{f,m}^{a-1} (1 - \theta_{f,m})^{b-1}
 \times \prod_{f=1}^F \theta_{f,u}^{a-1} (1 - \theta_{f,u})^{b-1} \times IG(\sigma^2|a_{\sigma},b_{\sigma})\times N\left((\bbeta,\alpha)'|\bzero,\bI\right)\nonumber\\
 &\times N(\mu_1|0,1) \times IG(\sigma_1^2|a_{\sigma_1},b_{\sigma_1}),
\end{align}
Summaries of these posterior distributions cannot be computed in closed form.  Thus, posterior computation proceeds through Markov chain Monte Carlo (MCMC) algorithms.  In each iteration, we update the outcome regression parameters using the current set of  model-determined links. We also re-estimate 
propensity scores based only 
on those records in File B  that have been linked to File A in that iteration. We re-estimate propensity scores since these records correspond to the linked dataset on which causal inference is performed. 
The full posterior conditionals can be found in the supplementary material.


For either model, we let the MCMC chain run  until apparent convergence ($2000$ iterations in our simulations) and discard an appropriate burn-in (the first $1500$ iterations in our simulations).  
Let $z_j^{(1)},  \dots, z_j^{(L)}$ be the $L$ post burn-in MCMC iterates of $z_j$, where $j=1, \dots, n_B$. For each $j$, we empirically estimate $P(z_j=q|-)$ using the proportion of post burn-in samples where $z_j$ takes the value $q$, i.e., $\hat{P}(z_j=q|-)=\#\{l:z_j^{(l)}=q\}/L$, for $q\in\mathcal{J}_j=\{1, \dots, n_A,n_A+j\}$.  The most likely link for record $j$ in File B is the record $q$ satisfying  $1\leq q^*=arg\max_{q\in\mathcal{J}_j}\hat{P}(z_j=q|-)\leq n_A$. We 
denote this record as $\hat{z}_j=q$. When $q^*=n_A+j$, we declare it most likely that record $j$ does not have a link in File A. The posterior distributions of each $z_j$ characterize the uncertainties associated with the links. 

For posterior inferences on causal effects, we define the average treatment effect for the linked cases, which we abbreviate as ATEL. 
\begin{align}\label{ATE}
\mbox{ATEL}=\frac{\sum\limits_{i\in\mathcal{A}}(y_i(1)-y_i(0))}{n_{AB}}=\frac{\sum\limits_{i\in\mathcal{A}}T_i}{n_{AB}},
\end{align}
where $\mathcal{A}=\{i:z_j=i,\:\mbox{for some j} \leq n_A\}$, and $n_{AB}$ denotes the cardinality of $\mathcal{A}$. In expectation, the ATEL equals  the usual average treatment effect for the records in File A when the linked records do not differ systematically from the full sample of File A; that is, linkages are independent of the potential outcomes.   As we do not know for certain which record pairs are true links, we estimate the posterior distribution of the ATEL by computing \eqref{ATE} in each iteration of the MCMC sampler.  


To draw posterior inferences on the ATEL,
define $y_{miss,i}=(1-w_i)y_i(1)+w_i y_i(0)$ as the counterfactual outcome for the $i$th record in File A, where $i=1, \dots, n_A$. At the $l$-th post burn-in iteration, we impute the counterfactual outcomes $y_{miss,i}^{(l)}$ for all linked individuals, i.e., all  $i\in\mathcal{A}^{(l)}=\{i:z_j^{(l)}=i,\:\mbox{for some j} \leq n_A\}$, from their posterior predictive distributions, 
\begin{align}\label{y:ppd}
p(y_{miss,i}|y_1, \dots, y_{n_A},z_j=i)=\int f_1(y_{miss,i}|w_i = 1-w_j, \bx_j,\btheta_c)p(\btheta_c|y_1, \dots, y_{n_B})d\btheta_c.
\end{align}
In \eqref{y:ppd}, we sample $y_{miss,i}^{(l)}$ using a treatment indicator that is opposite what is observed for its linked record, i.e., we set $w_i  = (1-w_j)$. 
We obtain the $l$-th post burn-in iterate for the ATEL using (\ref{ATE}) with $(y_i, y_{miss,i}^{(l)})$ over all  $i\in\mathcal{A}^{(l)}$. 

In the simulations in Section~\ref{secsim} and analyses in Section~\ref{rda}, we choose the values of the hyperparameters as $a_{\sigma}=1$, $b_{\sigma}=1$, $\alpha_{\pi}=1$, $\beta_{\pi}=1$,  $a_{\sigma_1}=1$, $b_{\sigma_1}=1$, $r_1=r_2=\delta_1=\delta_2=1$. 
Moderate perturbations of the hyperparameter values lead to practically indistinguishable results.

\section{Simulation Studies}\label{secsim}
We carry out simulation studies to assess the performance of the Bayesian hierarchical model, which for brevity we refer to as the joint model. We consider simulation scenarios in which we vary (a) the proportion of records in common between the two files and (b) the data generation model for the outcomes. 
Within these, we consider simulation scenarios with the correctly specified and a mis-specified outcome regression model. Finally, we present results from a simulation with missing outcome values. 

\subsection{Simulated Data Generation}\label{datgen}
We work with the  RLdata10000 data from the R package, \texttt{RecordLinkage} \citep{sariyar2010recordlinkage}. These data comprise an artificial population of $10000$ records with first names, last names, birth years and birth dates. Among these, there are $1000$ individuals whose values of these variables have been duplicated and then randomly perturbed, introducing errors into these potential linking variables. 

The RLdata10000 data do not include covariates, treatments, or outcomes.  Thus, we generate values of these for each of the $9000$ unique individuals in the RLdata10000 file.  For each individual $j$, we generate $p=2$ covariates, $x_{1,j}$ and  $x_{2,j}$, sampled  i.i.d. from standard normal distributions. 
We generate each individual's binary treatment assignment $w_j$ from a Bernoulli distribution with probability given by 
\begin{align}\label{true_prop}
e(\bx_j) = P(w_j = 1|\bx_j) = \frac{e^{\alpha_0 + \sum_{l=1}^p\alpha_l x_{l,j} }}{(1+e^{\alpha_0 + \sum_{l=1}^p\alpha_l x_{l,j} })},
\end{align}
where $(\alpha_0, \alpha_1, \alpha_2) = (1, 1.5, -1)$.
We generate each individual's outcome $y_j$ from 
\begin{align}\label{eq:true_causal}
y_j = m_1^0(\bx_j) +   m_2^0(\bx_j)w_j + \epsilon_{j}, \:\:\:\:\:\:\epsilon_{j} \sim N(0,1),
\end{align}
where the superscript $0$ indicates the true data generating mechanism.
We examine results for two choices of $(m_1^0, m_2^0)$.  The first uses  linear functions in the propensity score: $m_1^0(\bx_j)=1+2 e(\bx_j)$ and $m_2^0(\bx_j)=4$.  We call this Scheme L.  The second uses nonlinear functions in the propensity score:  $m_1^0(\bx_j)= 5 - 1.5 \:e(\bx_j)$ and $m_2^0(\bx_j)=\exp(-0.8 + 2.6 \: e(\bx_j))$.  We call this Scheme N. 

We construct File A and File B by putting subsets of these records into two files.  For any record, File A includes the outcome information, while File B includes the covariate and treatment information; both files include the imperfect linking variables. For ease of simulation, we set the sizes of File A and File B to be $n_A=n_B=1000$. 

In any simulation, we randomly sample a subset of the $1000$ individuals with duplicates. We put these records in File A and their duplicates in File B.  The number of these overlapping individuals is denoted by $O_{AB}$, 
which is varied to be 100, 500, or 900. For the remaining $(n_A-O_{AB})$ records in File A, we randomly choose $(n_A-O_{AB})$ records from the $8000$ individuals without duplicates, discarding their treatments and covariates and keeping their outcomes and the linking variables. To ensure that the non-overlapping records of File A and File B correspond to different individuals, we set aside these $(n_A-O_{AB})$ records from the $8000$ records. To add the remaining $(n_B-O_{AB})$ records to File B, we randomly choose $(n_B-O_{AB})$ records from the remaining $(8000-n_A+O_{AB})$ records, discarding their outcomes and keeping the treatments, covariates, and linking variables. 


For comparisons, we consider the performance of
two alternatives.  In the two-stage model,  we first link records using the posterior mode of each $z_j$ after fitting the bipartite Bayesian record linkage method as described in Section~\ref{RL}, without using the covariates, treatments, or outcomes.
We then perform causal inference  on the records linked in the initial step, i.e., the exercise is sequential as opposed to joint. Comparisons with this model reveal if the sharing  of information between the record linkage and outcome models offers any inferential advantages.  We also consider using the known links, that is, we make causal inferences with the true links. 
Although this approach is not feasible in practice, as one does not know the true links in genuine scenarios, we consider it a benchmark for the best we can do in these simulation scenarios.

For both the joint and two-stage models, we compare performance accuracy both in terms of record linkage and causal inference. For the former, we examine the positive predictive value (PPV) and the negative predictive value (NPV), defined as follows. 
Let $\hat{\bz}$ be the posterior mode of $\bz$. The PPV is the proportion of links that are actual matches and the NPV is the proportion of non-links that are actual non-matches. Let $\mathcal{A}_{1,j}=\{\hat{z}_j=z_j, z_j\leq n_A\}$ and $\mathcal{A}_{2,j}=\{\hat{z}_j=z_j, z_j= n_A+j\}$. Let $I(\mathcal{A}_{k,j})$ be the indicator function corresponding to set $\mathcal{A}_{k,j}$, $k=1,2$; $j=1, \dots, n_B$. The PPV and NPV are defined as $\sum_{j=1}^{n_B}I(\mathcal{A}_{1,j})/\sum_{j=1}^{n_B}I(z_j\leq n_A)$ and $\sum_{j=1}^{n_B}I(\mathcal{A}_{2,j})/\sum_{j=1}^{n_B}I(z_j= n_A+j)$, respectively. A perfect record linkage procedure would result in PPV=NPV=1. 

To assess the quality of causal inference for all three methods, we use the mean squared error (MSE) of the post burn-in causal effects, i.e., MSE $=\sum_{l=1}^L (\mbox{ATEL}^{(l)}-\mbox{ATEL}_0)^2/L,$ where $\mbox{ATEL}_0$ is the value of the ATEL computed using all true links and $\mbox{ATEL}^{(l)}$ is the $l$th post burn-in estimate of ATEL. 
We also examine the posterior distributions and 95\% credible intervals of the ATEL. 


\subsection{Results}\label{no_model_mis}
 We begin with results with no missing outcomes and with correct outcome model specifications. That is, we specify parametric or semi-parametric outcome regressions that match the choices of $m_1(\cdot)$ and $m_2(\cdot)$ in the data generation models in   
 Section \ref{datgen}. We estimate the joint model and the two-stage model using four linking variables: first name, last name, birth date and birth year. 
  

Table~\ref{Tab1} summarizes the PPV and NPV for the joint model and two-stage model, averaged over $20$ replications---enough to generate sufficiently small Monte Carlo errors---when using the correct outcome model specifications. 
For both data generation schemes, both the PPV and NPV of the joint model decrease as the  percentage of overlap between File A and File B decreases. The two-stage model follows a similar pattern. 
Comparing the two models, we see that the joint model tends to have a higher PPV than the two-stage model.  The improved performance of the joint model becomes increasingly apparent as the amount of overlap decreases. 
The joint model also tends to have a higher NPV than the two-stage model, although the differences are negligible in the scenario with 10\% overlap under Scheme N. 

\begin{table}[t]
\small{
\begin{center}
\begin{tabular}
[c]{cccccccc}
\hline
True & Percentage             & PPV                           & NPV                           & PPV                           & NPV \\
Model             & of Overlap & (\emph{Joint}) & (\emph{Joint}) & (\emph{Two-Stage}) & (\emph{Two-Stage})\\
\hline
                   & 90 & $0.99$ & $0.98$ & $0.97$ &  $0.95$\\
Scheme L & 50 & $0.99$  & $0.96$ & $0.94$ &  $0.92$\\
                  & 10 & $0.91$ & $0.93$ & $0.86$ &   $0.91$\\
\hline
                          & 90 & $0.99$ & $0.96$ & $0.97$  & $0.94$\\
Scheme N & 50 & $0.97$ & $0.96$ & $0.94$ &  $0.92$\\
                          & 10 & $0.96$ & $0.91$ & $0.87$ &  $0.91$\\
\hline
\end{tabular}
\caption{Positive predictive values (PPV) and negative predictive values (NPV) for the joint model and the two-stage model for different overlap levels when using correct outcome model specifications. All Monte Carlo standard errors are 0.004 or smaller. The known link model uses the true links and hence is not included.}\label{Tab1}
\end{center}
}
\end{table}


\begin{table}[h]
\small{
\begin{center}
\begin{tabular}
[c]{cccccccc}
\hline
True & Percentage             &  Joint     & Two-Stage  & Known Link  \\
Model              & of Overlap &  Model  & Model         & Model  \\
\hline
                   & 90 &  $0.02 \,(0.001)$ & $0.14 \,(0.003)$ & $0.01 \,(0.001)$\\
Scheme L & 50 &  $0.10 \,(0.002)$ & $7.94 \,(0.290)$ & $0.09 \,(0.001)$\\
                  & 10 & $1.78 \,(0.050)$  & $12.32 \,(0.340)$ & $0.68 \,(0.030)$\\
\hline
                          & 90 & $0.01 \,(0.001)$  &  $0.03 \,(0.001)$ &  $0.01 \,(0.001)$ \\
Scheme N & 50 & $0.03 \,(0.002)$ & $0.62 \,(0.040)$ & $0.02 \,(0.001)$\\
                          & 10 & $0.37 \,(0.036)$ & $1.64 \,(0.070)$ & $0.18 \,(0.002)$ \\
\hline
\end{tabular}
\caption{MSE of estimating the true causal effect ($ATEL_0$) for the joint model, the two-stage model, and using the known links for different overlap levels when using correct outcome model specifications. Monte Carlo standard errors are in parentheses.}\label{Tab2}
\end{center}
}
\end{table}

The improvements in the linkages when using the joint model has benefits for the estimation of the causal effect. As evident in Table \ref{Tab2},   
the joint model performs significantly better on MSE than the two-stage model. The performance gap becomes more substantial as the percentage of overlap decreases, especially under Scheme L. Notably, the results from the joint model are similar to those from the gold-standard Known Link model in the 50\% and 90\% overlap scenarios.

We next examine performance when the outcome model does not exactly match the data generating model.  In particular, we fit an outcome regression that is linear in the propensity score even though the outcomes are generated using Scheme N; and, we fit an outcome regression that uses the penalized splines even though the outcomes are generated using Scheme L.  Here, we only consider the scenario with 90\% overlap, which gives both methods the best chance to perform well. 
As evident in Table~\ref{Tab_mis}, not surprisingly performances of both models deteriorate substantially compared to the results in Table~\ref{Tab1} and Table~\ref{Tab2}. We are imputing potential outcomes from mis-specified models, after all.
When fitting the semi-parametric model to data generated under Scheme L, we observe higher PPV and NPV, as well as a lower MSE, for the joint model compared to the two-stage model. This is also the case when fitting the parametric model to data generated under Scheme N; however, in this scenario the differences are practically modest.  Taken together, these results suggest that, even with model mis-specification, it may be advantageous to use the joint model over the two-stage model.

\begin{table}[th]
\small{
\begin{center}
\begin{tabular}
[c]{cc|cc|cc|cc}
\cline{1-8}
True & Fitted & \multicolumn{2}{c}{PPV} & \multicolumn{2}{c}{NPV} & \multicolumn{2}{c}{MSE}\\
\cline{3-8}
Model & Model & Joint & Two-Stage & Joint & Two-Stage & Joint & Two-Stage\\
\hline
Scheme L & Splines & $0.99$ & $0.97$ & $0.96$ & $0.94$ & $0.79\,(0.038)$ & $1.28\,(0.056)$\\
Scheme N & Linear & $0.99$ & $0.97$ & $0.97$ & $0.94$ & $0.56\,(0.020)$ & $0.62\,(0.019)$\\
\hline
\end{tabular}
\caption{PPV, NPV, and MSE for the joint model and the two-stage models under model mis-specification. Monte Carlo standard errors for MSE values are presented in parentheses.  All Monte Carlo standard errors for PPV and NPV are $0.005$ or smaller. Results based on 90\% overlap of records between File A and File B.
 }\label{Tab_mis}
\end{center}
}
\end{table}


Finally, we  examine the performance of the joint model and two-stage model in the presence of missing outcomes in File A. We blank either $5\%$ or $10\%$ of the values of $y_i$ in File A using a missing completely at random mechanism.  We examine the cases of correct model specifications with $90\%$ overlap of records between File A and File B. 
To handle the missing values in the joint model, we sample the missing observations from their posterior predictive distributions in each MCMC iteration. For the two-stage model, based on the posterior mode $\hat{z}_j$ of $z_j$, we impute the missing values from their posterior predictive distributions after fitting the outcome model on the linked dataset.    

\begin{table}[t]
\small{
\begin{center}
\begin{tabular}
[c]{cc|cc|cc|cc}
\cline{1-8}
True & Missing  & \multicolumn{2}{c}{PPV} & \multicolumn{2}{c}{NPV} & \multicolumn{2}{c}{MSE}\\
\cline{3-8}
Model  & \% & Joint & Two-Stage & Joint & Two-Stage & Joint & Two-Stage\\
\hline
Scheme L & 5\% & $0.99$ & $0.97$ & $0.97$ & $0.94$ & $0.11\,(0.004)$ & $0.16\,(0.005)$\\
Scheme L & 10\% & $0.98$ & $0.97$ & $0.97$ & $0.94$ & $0.16\,(0.005)$ & $0.19\,(0.006)$\\
\hline
Scheme N & 5\% & $0.99$ & $0.97$ & $0.97$ & $0.95$  & $0.02\,(0.001)$ & $0.05\,(0.001)$\\
Scheme N & 10\% & $0.99$  & $0.97$ & $0.96$ & $0.94$ & $0.05\,(0.001)$ & $0.07\,(0.001)$\\
\hline
\end{tabular}
\caption{PPV, NPV, and MSE for the joint and the two-stage models with 5\% and 10\% missing outcomes in File A. All Monte Carlo standard errors for PPV and NPV are $0.004$ or smaller.
Results based on 90\% overlap of records between File A and File B.}\label{Tab_missing_data}
\end{center}
}
\end{table}

Table~\ref{Tab_missing_data} summarizes results over 20 independent simulation runs.  The performance of joint model worsens as the percentage of missing data increases, although not by much in these scenarios. 
A similar trend is observed for the two-stage model. 
We continue to see advantages of the joint model over the two-stage model.

In the supplementary material, we describe results from additional simulation scenarios.  In particular, we 
find that the relative performances of the joint model and two-stage model remain qualitatively unchanged when using correlated (rather than independent) covariates. We also lower the signal to noise ratio by increasing the regression variance.  Not surprisingly, the performance gap between the joint model and the two-stage model closes as the variance increases.

\section{Causal Study of Debit Cards}\label{rda}

The past few decades have seen a steadily increasing global trend in the use of non cash payment instruments like credit, debit and prepaid cards. 
\cite{thaler1985mental} and \cite{thaler1999mental}  argue that the form of payment instruments can have a significant impact on consumer decisions via mental accounting, a set of cognitive operations used by individuals and households to keep track of financial activities. Indeed, there is evidence that consumers who have cards would spend more than ones who do not \citep{cole1998identifying}. A comprehensive causal study carried out by \cite{mercatanti2014debit} in this regard focuses on the effect of debit cards on spending. \cite{mercatanti2014debit} argue that debit cards, unlike credit cards, do not allow consumers to incorporate additional long-term sources of funds in their spending decisions, thus eliminating any confounding intertemporal reallocations of wealth from the psychological
effects on spending \citep{soman2001effects}, and hence are more appropriate to look at for this kind of a causal study.

With this background in mind, we use an observational study of the causal effect of possession of debit cards on household consumption to illustrate the Bayesian hierarchical model for causal inference and record linkage, as we now describe. 

\subsection{Data Description and Background}\label{rda1}

We use data from the Italy Survey on Household Income and Wealth (SHIW).
The SHIW is a nationally representative survey, run by the Bank of Italy once in every 2 years since 1965, with the only exception being that the 1997 survey was delayed to 1998. The purpose of this survey is to collect information on several aspects of Italian households' economic and financial behavior.  Since the data contain information related to household characteristics, spending and payment instruments, the SHIW can provide a useful opportunity to evaluate the causal effect of debit card possession on spending in Italian households.

We link two files comprising data collected during the years 1995 and 1998. A number of the same households participated in both years. In particular, our target population is the set of households having at least one current bank account but no debit cards before 1995. The treatment $w = 1$ if the household (all members combined) possesses one and only one debit card at 1998, and $w=0$ if the household does not possess any debit cards at 1998. Households with more than one debit card are excluded from our sample. Here, it may be mentioned that ideally, an analysis with units being individuals that possess debit cards should be carried out, because debit cards are typically issued to individuals. But the SHIW survey only has this information at the household level. Our strategy to limit the sample of treated units to households possessing only one debit card ensures that a possible effect on household spending will be due to a certain individual possessing this card. Though we do not have exact information on the ownership of the card, we make the (reasonable) assumption that the head of the household has possession of the sole debit card. 

The outcome on which we evaluate the treatment effect is the monthly average spending of the household on all consumer goods, measured in the latter survey (1998). For data quality control, we delete $15$ observations which have either negative values of the outcome (monthly spending) or unusually high ratios (greater than 5 and going up to 900) of monthly spending to monthly income. Upon implementing such data quality control measures, the data file corresponding to 1995 contains 589 observations with information on the treatment (debit card possession) and covariates, while the data file corresponding to 1998 (3919 observations) contains information on the outcome (monthly average household spending). 

Both files contain a common set of imperfect linking variables, including 
the geographical area of residence of the household, the number of inhabitants in the town of the household, and the gender, birth year, marital status, region of birth and highest educational qualification of the head of the household. Fortunately, we also have a unique ID that we can use to perfectly link households across years. 
We use this ID variable to assess how well our model has linked observations in the two files, based on the other imperfect linking variables noted above. Using the unique matching ID, we observe that the file contains 191 observations in the treatment group (who possess a debit card) and the other 398 observations in the control group. An initial check on the spending distribution for the treatment and control groups (see Figure~\ref{Adj1}) hints at a positive effect of  acquiring a debit card on household spending. 


The covariates (possible confounders) we consider in this study are all measured in the initial survey (1995), and consist of the monthly average spending of the household on consumer goods in the initial survey year (lagged outcome), the net wealth of the household, the household net disposable income, the monthly average cash inventory held by the household, the average interest rate and the number of banks in the municipality where the household is located. The choice of these confounders appears to satisfy the strong ignorability condition, as we discuss in the online supplement.
The lagged outcome is generally indicated in the economic literature as a fundamental confounder  \citep{angrist2009instrumental, frolich2019impact}. The cash inventory held by the household was introduced in the specific context by \cite{mercatanti2014debit}. The net wealth and the net disposable income are important indicators of the household economic condition.  The last two covariates have been suggested by \cite{attanasio2002demand}, who have shown in non-causal contexts that the interest rate and the number of banks in the municipality where the household lives had a significant contribution to the probability of acquiring a debit card in Italy. Moreover, the number of banks is a good indicator of the size of the municipality. 

\subsection{Results}\label{rda2}
We implement the joint model with the semi-parametric outcome regression specification discussed in Section \ref{sec:causla_model}. 
We also include the two-stage model and the results using the known-links for comparisons. 
We use the same prior hyperparameter values as in the simulation studies; moderate perturbation of them leads to practically indistinguishable results. We let the MCMC chain run for 2000 iterations and discard the first 1500 as burn-in, and draw inferences on both the ATEL and record linkage based on the post burn-in iterates. 

\begin{table}[t]
\begin{center}
\begin{tabular}
[c]{l|r|r|rrr}
\cline{1-6}
Fitted Model & PPV & NPV & \multicolumn{3}{c}{ATEL}\\
\cline{4-6}
 &  &  & 2.5\% & 50\% & 97.5\% \\
\hline
Known-Link & -- & -- & 127.49 & 233.14  & 334.67 \\
Joint & 0.876 & 0.979 & 108.46 & 258.57  & 412.44  \\
Two-Stage & 0.847 & 0.881 & 84.08 & 193.28 & 306.16 \\
\hline
\end{tabular}
\caption{PPV and NPV for linking the 1995 and 1998 files in the SHIW causal study. Also included are the 2.5\%, 50\% and 97.5\% quantiles of the posterior distribution of the ATEL (in thousand Italian Liras) for all methods.}\label{Tab_real_data}
\end{center}
\end{table}

Table \ref{Tab_real_data} presents the PPV and NPV values, along with the posterior median and 95\% credible intervals of the estimated ATEL (in thousand Italian Liras) for all models. Consistent with the simulation results, the joint model offers a noticeably better PPV and NPV than the two-stage model. 
Using the results from the known-links as a benchmark, we find that the posterior inferences for the joint model seem more plausible than those from the two-stage model.  First, the posterior medians for the joint model and the known link model are more similar to one another than are the the posterior medians for the known-link model and two-stage model.  Second, the 95\% credible interval for the joint model is wider than the interval for the known-links model, which is sensible in that it reflects additional uncertainty from imperfect linkages.  On the other hand, the 95\% credible interval from the two-stage model actually is practically the same length as the interval for the known-links model, effectively portraying no propagation of uncertainty from imprecise linkages.  


Figure~\ref{Adj2} displays the posterior distributions of the ATEL.  The results suggest that, on average, the effect of possession of a single debit card for a household leads to more monthly consumption than households that do not  possess any debit card during the study period. 
Our analysis largely eliminates any potential confounding effect of intertemporal reallocation of wealth, since debit cards do not allow for long-term fund sources \citep{soman2002effect}. Hence, the significant estimated effects of debit card possession on spending may be attributed to psychological reasons (increased perceived amount of money) \citep{soman2001effects} and easier accessibility to financial resources \citep{morewedge2007unfixed}. The estimated ATEL is higher than the ATT (the Average Treatment effect on the Treated) ($\sim$ 200 thousand Italian Liras), in \cite{mercatanti2014debit}. This result is interpretable in the light of some recent economic models for the use of debit cards (e.g., see \citealp{kim2010modelcard} and references therein), which imply that the poor adopt  debit cards later than the rest of the population.
This is confirmed by \cite{mercatanti2014debit} who show that households with debit cards generally have higher levels of income, wealth and education of the members in comparison with households without debit cards. Therefore, our estimated ATEL values indicate larger psychological  effects on spending for people in disadvantageous social and economic conditions.


\begin{figure}
    \begin{center}
 \subfigure[]{\includegraphics[width=9cm,height=7cm]{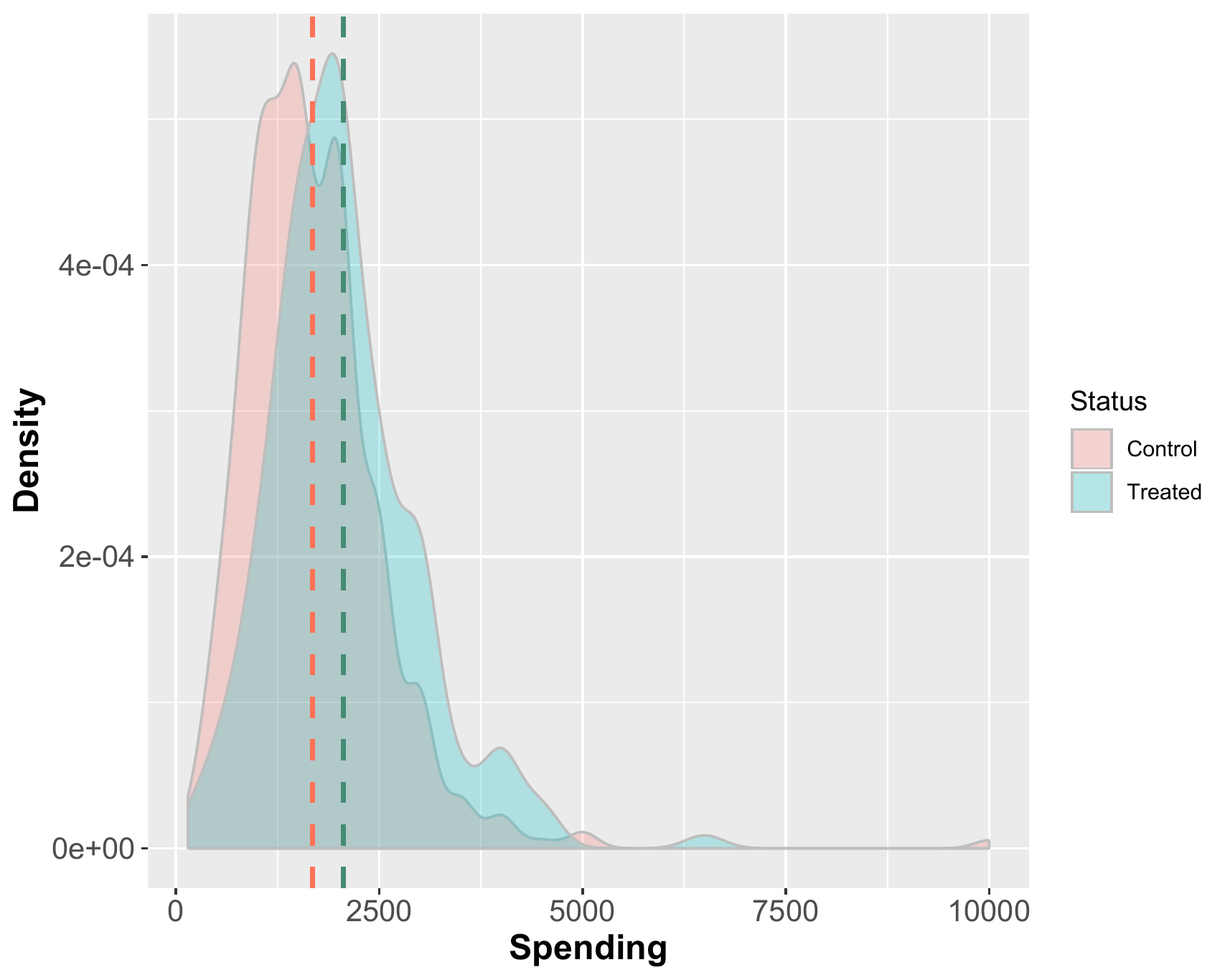}\label{Adj1}}
   \subfigure[]{\includegraphics[width=8cm,height=8cm]{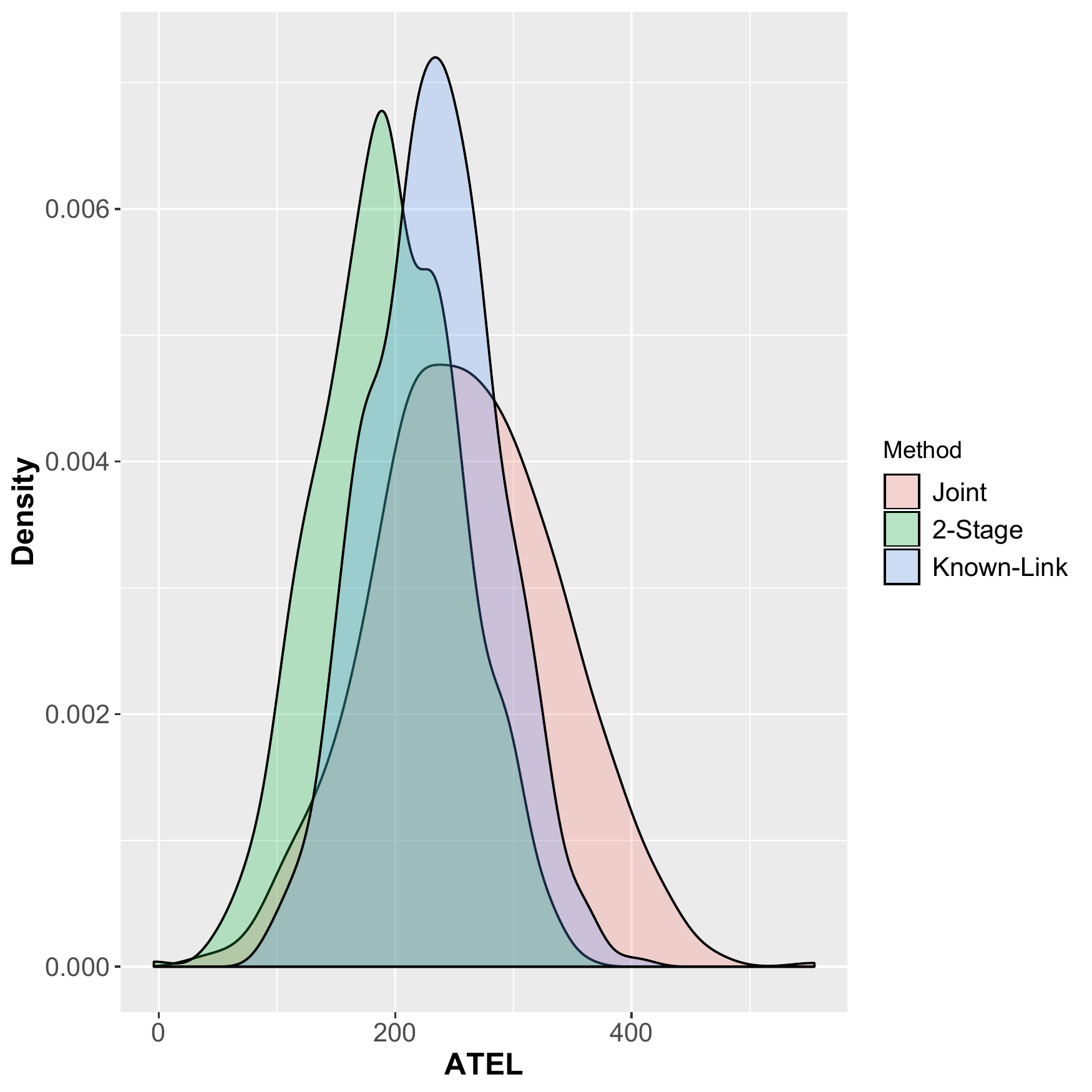}\label{Adj2}}
     \end{center}
   \caption{Figure \ref{Adj1} shows the average per month spending distribution of the treatment and control groups. The vertical lines indicate the means of the two distributions. Figure \ref{Adj2} shows the distribution of the ATEL for the joint, two-stage and the known-link models. The numbers are in per thousand Italian Liras.} \label{fig:ATE_spending}
 \end{figure}

\section{Discussion and Future Work}\label{conclusion}

The Bayesian approach to causal inference and record linkage offers interesting future directions. For example, many data applications have predictors and treatment status residing in different files. This requires significant modifications of the approach presented here, as one needs a model for the covariates as well as the outcomes. Another important future direction is to extend this approach to other flexible outcome models.

We conclude with a connection to the philosophy of causal inference. Performing causal inference and record linkage simultaneously allows the values of the outcome variables to influence which records are used in the causal estimator.  This is in conflict with the often followed advice that the design of the observational study should proceed separately from the analysis \citep{imbens2015causal}. As suggested by \cite{wortman2018simultaneous}, if one seeks the potential gains in accuracy from using the relationships among the variables, this is the price to pay for working with imperfect linkages.



\bibliographystyle{natbib}
\bibliography{reference}


\end{document}